\documentclass[draftclsnofoot, onecolumn]{IEEEtran}

\usepackage{amsfonts, amssymb, graphicx, amsthm}
\usepackage{mathtools,soul}
\usepackage[normalem]{ulem}
\usepackage{bm}
\usepackage{array}
\usepackage{xcolor}

\usepackage{microtype}

\usepackage[capitalize]{cleveref}

\crefname{section}{section}{sections}
\crefname{subsection}{subsection}{subsections}
\Crefname{section}{Section}{Sections}
\Crefname{subsection}{Subsection}{Subsections}

\crefname{figure}{Figure}{Figures}

\theoremstyle{plain}
\newtheorem{theorem}{Theorem}[section]

\theoremstyle{definition}
\newtheorem{definition}[theorem]{Definition}

\theoremstyle{remark}
\newtheorem{remark}[theorem]{Remark}
\newtheorem{example}[theorem]{Example}

\newcommand{\security}{\mathrm{sec}}
\newcommand{\F}{\mathbb{F}}
\newcommand{\N}{\mathbb{N}}
\newcommand{\R}{\mathbb{R}}
\newcommand{\C}{\mathbb{C}}

\renewcommand{\phi}{\varphi}
\renewcommand{\P}{\mathbb{P}}
\renewcommand{\hat}{\widehat}
\renewcommand{\tilde}{\widetilde}

\DeclareMathOperator{\chr}{char}
\DeclareMathOperator{\ev}{ev}
\DeclareMathOperator{\Div}{Div}
\DeclareMathOperator{\spn}{span}
\DeclareMathOperator{\supp}{supp}
\DeclareMathOperator{\RS}{RS}

\newcolumntype{L}{>{$}l<{$}}
\newcolumntype{C}{>{$}c<{$}}
\newcolumntype{R}{>{$}r<{$}}

\usepackage{enumitem}
\setlist[enumerate]{leftmargin=.5in}
\setlist[itemize]{leftmargin=.5in}

\title{Algebraic Geometry Codes for Secure Distributed Matrix Multiplication}

\author{
\IEEEauthorblockN{
    Okko~Makkonen, Elif~Saçıkara and Camilla~Hollanti
}

\IEEEauthorblockA{
    Department of Mathematics and Systems Analysis \\
    Aalto University, Finland \\
    \texttt{\{okko.makkonen, elif.sacikara, camilla.hollanti\}@aalto.fi}
}
}
\IEEEoverridecommandlockouts

\begin{document}

\maketitle

\begin{abstract}
In this paper, we propose a novel construction for secure distributed matrix multiplication (SDMM) based on algebraic geometry (AG) codes, which we call the PoleGap SDMM scheme. The proposed construction is inspired by the GASP code, where so-called gaps in a certain polynomial are utilized to achieve higher communication rates. Our construction considers the gaps in a Weierstrass semigroup of a rational place in an algebraic function field to achieve a similar increase in the rate. This construction shows that there is potential in utilizing AG codes and their subcodes in SDMM since we demonstrate a better performance compared to state-of-the-art schemes in some parameter regimes.
\end{abstract}

\section{Introduction}

Secure distributed matrix multiplication (SDMM) was first introduced by Chang and Tandon in \cite{chang2018capacity} as a way to utilize a distributed system to compute the product of two matrices without revealing information about the matrices to the workers. SDMM can be used to outsource the computation of large matrix products in data science and machine learning when the privacy of the data in the matrices is vital. Several schemes have been introduced, including those in \cite{d2020gasp, d2021degree, kakar2019capacity, aliasgari2020private} and more recently in \cite{lopez2022secure, hasircioglu2022bivariate, mital2022secure}. Additionally, SDMM schemes based on the so-called grid partition have been presented in \cite{machado2022root, karpuk2023modular, byrne2023straggler}. Many of the schemes presented in the literature follow a similar structure, where the encoding is done with some linear codes. A general framework, titled \emph{linear SDMM}, was introduced in \cite{makkonen2022general}. There, the codes that are used in the encoding, as well as their star product, are studied to give some fundamental bounds on such schemes.

An important class in linear codes are \emph{algebraic geometry (AG) codes}, which are constructed from projective smooth irreducible algebraic curves, or equivalently, \emph{algebraic function fields}. The first construction was introduced by Goppa in \cite{goppa1977codes} as a generalization of Reed--Solomon codes by considering \emph{Riemann--Roch spaces} instead of spaces of bounded degree polynomials. Since then, these codes have been shown to have nice coding-theoretic properties derived from the structure of algebraic function fields. Another breakthrough construction was given by Tsfasman, Vl{\u{a}}dutx and Zink in \cite{tsfasman1982modular}. Their construction was based on towers of algebraic function fields that behave asymptotically well and yield better parameters than the Gilbert--Varshamov bound.

For an AG code construction, one can work with a basis of a Riemann--Roch space to provide a generator matrix of the code. However, computing an explicit basis of a Riemann--Roch space is difficult in general \cite{hess2002computing}. Correspondingly, a similar difficulty can also be observed in the determination of \emph{Weierstrass semigroups}, which can be used to retrieve some information about the bases of Riemann--Roch spaces of specific divisors \cite{garcia1986weierstrass, castellanos2016one}. In this paper, we consider the Weierstrass semigroups to choose bases for subspaces of a Riemann--Roch space.

The theory of algebraic geometry codes has also been used for applications such as locally recoverable codes (LRCs) and code-based cryptography. For instance, the importance of AG codes for LRCs was shown in \cite{barg2017locally} as a generalization of the case where Reed--Solomon codes were used in \cite{tamo2014family}. Another well-known application of AG codes was considered by McEliece in \cite{mceliece1978public} where a public key cryptosystem based on \emph{Goppa codes} was introduced. See \cite{couvreur2017cryptanalysis, couvreur2020algebraic} for surveys on the applications of algebraic geometry codes.

\subsection{Contributions and Related Work}

In this paper, we utilize algebraic geometry codes to construct a linear SDMM scheme, called the PoleGap SDMM scheme, that outperforms the state-of-the-art schemes in terms of the download rate for some parameter choices. Our construction utilizes a similar strategy as the Gap Additive Secure Polynomial (GASP) scheme, introduced in \cite{d2020gasp, d2021degree}, where the ``gaps'' in a certain polynomial are utilized to decrease the download cost. Such polynomials are designed using the help of a \emph{degree table}. Similarly, we utilize the gaps in a Weierstrass semigroup of an algebraic function field, yielding a \emph{pole number table}. The response code of our construction will be an AG code, which means that suitable error correction algorithms from the literature can be utilized to provide robustness against Byzantine workers. This is in contrast to previous outer product partitioning schemes, where the response codes have not had similar structure. While the improvement in the rate may seem modest, keeping in mind that the computations in, \emph{e.g.}, machine learning, are massive, even small relative gains can translate into significant absolute savings in energy consumption and expenses. Furthermore, we hope that these first results serve as a proof-of-concept and prompt further improved constructions and new research directions.

There have been some generalizations to the schemes based on polynomial encoding, including those in \cite{hasircioglu2022bivariate}, but to the best of our knowledge, this paper is the first to explicitly utilize algebraic geometry codes in SDMM. However, our preprint \cite{hollanti2023algebraic} was soon followed by \cite{machado2023hera}, where an inner product partitioning scheme is constructed by utilizing Hermitian codes, with the aim of reducing the field size. Their work differs from ours in that they consider a different matrix partitioning as well as utilizing a field extension. The reason for restricting ourselves to function fields that do not require extension fields is explained in \cref{rmk:field_size}.

\subsection{Organization}

The structure of our paper is as follows. In \cref{sec:preliminaries} we recall necessary background from algebraic geometry codes and linear SDMM. In \cref{sec:AG_codes_in_SDMM} we introduce how to use algebraic geometry codes in SDMM by introducing the pole number table. In \cref{sec:construction} we present our construction of the PoleGap scheme. Finally, in \cref{sec:comparison} we compare our construction to some schemes from the literature.

\section{Preliminaries}\label{sec:preliminaries}

 Let $\F_q$ denote the finite field with $q$ elements, $\F_q^\times$ the group of units in $\F_q$, $\N_0 = \{0, 1, 2, \dots\}$ and $[n] = \{1, 2, \dots, n\}$. In this section, we present necessary background on algebraic function fields and algebraic geometry codes by mainly following  \cite{stichtenoth2009algebraic}.

\subsection{Algebraic Function Fields}\label{sec:algebraic_function_fields}

Let $K$ be a field. For the applications in algebraic coding theory and cryptography, the field $K$ is chosen as a finite field $\F_q$, but the results in this section hold for arbitrary base fields. Let $K[x]$ denote the polynomial ring in one variable. Recall that we define $K(x)$ as the field of fractions of $K[x]$. That is, 
\begin{equation*}
    K(x) \coloneqq \left\{ \frac{f(x)}{g(x)} \colon f(x), g(x) \in K[x] \text{ with } g(x) \neq 0  \right\},
\end{equation*}
which is called \emph{the rational function field} over $K$. One can easily see that $K(x)$ is indeed a field and an infinite transcendental extension over $K$. 

A field $F$ is said to be a function field over $K$, denoted by $F/K$, if it is a finite extension of $K(x)$. We often consider simple algebraic extensions of $K(x)$, which are generated by an irreducible polynomial, \emph{i.e.}, $F = K(x,y)$, where $y \in F$ is a root of some irreducible polynomial $\phi(T) \in K(x)[T]$.

It is well-known that such extensions, like in number theory or algebraic geometry, are studied in terms of certain \emph{local} components. In this paper, we use the standard definitions and notations as in \cite{stichtenoth2009algebraic}. Briefly, let $P$ be a \emph{place} in the set of places, denoted by $\P_F$, in a function field $F/K$. Recall that $P$ is called a \emph{rational place} if $\deg(P) \coloneqq [F_P : K] = 1$, where $F_P$ is the \emph{residue class field} of $P$. We denote the set of rational places by $\P_F^1$. For a given place $P$, $v_P$ denotes the corresponding \emph{discrete valuation map}. For a function $f \in F$ and place $P \in \P_F$, we say $P$ is a \emph{zero} (resp., \emph{pole}) of $f$ if $v_P(f) > 0$ (resp.,  $v_P(f) < 0$). Here, let us note that the discrete valuation map $v_P$ measures the multiplicity of the function $f$ at $P \in \P_F$. The places of the function field $F/K$ and $K(x)/K$ are related in the following way. Each place $P' \in \P_F$ contains a unique place $P \in \P_{K(x)}$ and there are at most $[F : K(x)]$ places $P' \in \P_F$ containing $P \in \P_{K(x)}$ by \cite[Theorem~3.1.11]{stichtenoth2009algebraic}.

A \emph{divisor} is defined as a formal sum of places and represented by $G \coloneqq \sum_{P \in \P_F} n_P P$ with finitely many nonzero integers $n_P$. The support of $G$ is the set of places with nonzero coefficients, \emph{i.e.}, $\supp(G) \coloneqq \{P \in \P_F \colon n_P \neq 0\}$. Also, we define \emph{the degree of a divisor} as $\deg(G) \coloneqq \sum_{P \in \P_F} n_P \deg(P)$. We denote the divisor group of $F$, \emph{i.e.}, the set containing all divisors, by $\Div(F)$. Since every function in $F$ has finitely many zeros and poles, we can introduce the zero and pole divisors of a function $f \in F$ by $(f)_0 \coloneqq \sum_{v_P(f) > 0} v_P(f) P$ and $(f)_\infty \coloneqq -\sum_{v_P(f) < 0} v_P(f) P$, respectively. The principal divisor of $f \in F$ is then $(f) = (f)_0 - (f)_\infty$.

Recall that a \emph{Riemann--Roch space} of a divisor $G \in \Div(F)$ is defined as
\begin{equation*}
    \mathcal{L}(G) \coloneqq \{x \in F \colon (x) + G \geq 0\} \cup \{0\}. 
\end{equation*}
This set is a finite-dimensional vector space over $K$, whose dimension is denoted by $\ell(G)$. However, computing the dimension or giving an explicit basis is difficult and leads to the definition of the \emph{genus} of a function field. Let us also recall that if $\deg(G) < 0$, then $\ell(G) = 0$. If $F/K$ is a function field of genus $g$ and $G \in \Div(F)$ with $\deg(G) \geq 2g - 1$, then $\ell(G) = \deg(G) + 1 - g$. 

In this paper we are interested in the Riemann--Roch spaces of divisors of the form $G = kP$ for some rational place $P$. Such Riemann--Roch spaces can be studied using the \emph{Weierstrass semigroup} $W(P)$. For a given place $P \in \P_F$, we define
\begin{equation*}
    W(P) \coloneqq \{n \geq 0 \mid \exists x \in F \colon (x)_\infty = nP \}.
\end{equation*}
The elements in $W(P)$ are said to be \emph{pole numbers} of $P$, while the elements in $\N_0 \setminus W(P)$ are called \emph{gaps}. It is easy to see that $W(P)$ is a semigroup under addition. The following theorem provides partial information on the structure of the Weierstrass semigroup.

\begin{theorem}[{Weierstrass Gap Theorem, \cite[Theorem 1.6.8]{stichtenoth2009algebraic}}]
Let $F/K$ be a function field of genus $g$ and $P \in \P_F^1$. Then $W(P) = \N_0 \setminus \mathcal{G}$, where $\lvert \mathcal{G} \rvert = g$. Furthermore, if $g > 0$, then $\mathcal{G} = \{i_1, \dots, i_g\}$ with $1 = i_1 < i_2 < \dots < i_g \leq 2g - 1$.
\end{theorem}

\subsection{Kummer Extensions}\label{sec:Kummer_extensions}

In this section we consider certain special cases of \emph{Kummer extensions} described in \cite[Proposition 3.7.3]{stichtenoth2009algebraic}. As a summary, we state the following theorem about the structures needed in \cref{sec:AG_codes_in_SDMM}.

\begin{theorem}\label{thm:Kummer_special}
Let $K$ be a field with odd characteristic. Let $F = K(x,y)$ with $y^2 = f(x)$, where
\begin{equation*}
    f(x) = \prod_{i=1}^{d} (x - \alpha_i)
\end{equation*}
with distinct points $\alpha_1, \dots, \alpha_d \in K$ and an odd integer $d$. Then the following items hold.
\begin{enumerate}
    \item The function field $F/K$ has genus $g = \frac{d - 1}{2}$.
    \item We have $(x)_\infty = 2P_\infty$ and $(y)_\infty = dP_\infty$, where $P_\infty$ is the unique place at infinity.
\end{enumerate} 
\end{theorem}

For these types of extensions, the following theorem describes the Weierstrass semigroup of $P_\infty \in \P_F$.

\begin{theorem}[{\cite[Theorem 3.2]{castellanos2016one}}]\label{thm:Weierstrass_semi_structure}
Let $F/K$ be a function field described in \cref{thm:Kummer_special}. Then the Weierstrass semigroup of $P_\infty$ is an additive semigroup generated by $2$ and $d$.
\end{theorem}

Now, we consider the two following cases depending on the degree $d$ of $f(x)$ in \cref{thm:Kummer_special}. 

\begin{example}[Elliptic Function Field]\label{ex:elliptic}
Let $d = 3$ in \cref{thm:Kummer_special}. Namely, let $f(x)$ be a square-free polynomial of degree $3$ in $K[x]$ and suppose that $\chr(K)$ is odd. Then $F = K(x, y)$ with $y^2 = f(x)$ is an elliptic function field with genus $g = 1$. The place $P_\infty$ is the unique common pole of $x$ and $y$ in $F$. Indeed, $(x)_\infty = 2P_\infty$ and $(y)_\infty = 3P_\infty$, and the Weierstrass semigroup of $P_\infty$ is $W(P_\infty) = \N_0 \setminus \{1\}$.

For any nonnegative integer $k \geq 2g - 1 = 1$, the dimension of $\mathcal{L}(kP_\infty)$ is $\ell(kP_\infty) = k$. Note that the functions $x^iy^j$ with nonnegative integers $i, j$ such that $2i + 3j \leq k$ forms a basis for $\mathcal{L}(kP_\infty)$.
\end{example}

\begin{example}[Hyperelliptic Function Field]\label{ex:hyperelliptic}
Again, let $\chr(K)$ be odd and let $d \geq 5$ be an odd integer in \cref{thm:Kummer_special}. Let $f(x)$ be a square-free polynomial of degree $d \geq 5$ in $K[x]$. Then $F = K(x, y)$ with $y^2 = f(x)$ is a hyperelliptic function field with genus $g = \frac{d - 1}{2}$. The place $P_\infty$ is the unique common pole of $x$ and $y$ in $F$ and the pole divisors are $(x)_\infty = 2P_\infty$ and $(y)_\infty = dP_\infty$. The Weierstrass semigroup of $P_\infty$ is $W(P_\infty) = \N_0 \setminus \{1, 3, \dots, 2g - 1\}$. For any nonnegative integer $k \geq 2g - 1$ the dimension of $\mathcal{L}(kP_\infty)$ is $\ell(kP_\infty) = k + 1 - g$.
\end{example}

Finally, in our construction we need an estimation for the number of rational places $N \coloneqq \lvert \P_F^1 \rvert$ when $K = \F_q$ is a finite field. Clearly, in the rational function field $\F_q(x)$ there are exactly $q + 1$ rational places. In general, it is difficult to determine $N$ in $F/\F_q$. Here we consider the estimate
\begin{equation*}
    \lvert N - (q + 1) \rvert \leq 2g\sqrt{q},
\end{equation*}
which is known as the \emph{Hasse--Weil bound}.

\subsection{Algebraic Geometry Codes}\label{sec:AG_codes}

A linear code $\mathcal{C}$ over a finite field $\F_q$ is defined as a $k$-dimensional subspace of $\F_q^n$ endowed with the Hamming metric. Such a code is denoted by its parameters $[n, k, d]_q$, where $d$ is the \emph{minimum distance} of the code, \emph{i.e.}, the minimal number of nonzero coordinates in a nonzero codeword. These parameters are related through the well-known \emph{Singleton bound} $k + d \leq n + 1$. Codes that achieve this bound with equality are called \emph{maximum distance separable (MDS)}. Another description of linear codes can be given as the image of an $\F_q$-linear map between two finite-dimensional vector spaces over $\F_q$.

\begin{example}[Reed--Solomon Codes]
Let $\F_q[x]_{< k} \coloneqq \{f \in \F_q[x] \mid \deg(f) < k \}$ be the set of polynomials of degree at most $k - 1$ in $\F_q[x]$ where $0 \leq k \leq n$. Consider $n$ distinct elements $\alpha = \{\alpha_1, \dots, \alpha_n\}$ in $\F_q$. The evaluation map $\ev_\alpha \colon \F_q[x]_{< k} \to \F_q^n$ defined by
\begin{equation*}
    f \mapsto (f(\alpha_1), \dots, f(\alpha_n))
\end{equation*}
is an $\F_q$-linear map with $\ker(\ev_\alpha) = \{0\}$. Then
\begin{equation*}
    \RS_k(\alpha) \coloneqq \ev_\alpha(\F_q[x]_{< k })
\end{equation*}
is an $[n, k, d]_q$ linear code and is called a Reed--Solomon (RS) code. It is well-known that Reed--Solomon codes are MDS, \emph{i.e.}, their minimum distance is $d = n - k + 1$.
\end{example}

One of the breakthrough generalizations of Reed--Solomon constructions was given by Goppa in the 1970s \cite{goppa1977codes}. Like in the case of Reed--Solomon codes, it can be easily observed that the parameters and properties of such constructions can be determined and studied by the notions in algebraic function fields. More explicitly, we first fix an algebraic function field $F/\F_q$. To define a linear map, we consider a Riemann--Roch space $\mathcal{L}(G)$ of a divisor $G \in \Div(F)$ as the domain vector space and $n$ distinct rational places $\mathcal{P} = \{P_1, \dots, P_n\}$ that are not in the support of $G$. Then a linear map can be defined by evaluating functions $f \in \mathcal{L}(G)$ at the given places. Indeed, the map $\ev_\mathcal{P} \colon \mathcal{L}(G) \to \F_q^n$ defined by
\begin{equation*}
    f \mapsto (f(P_1), \dots , f(P_n))
\end{equation*}
is well-defined as the $P_i$'s are not in the support of $G$. We define an \emph{algebraic geometry (AG) code} as the image of such an evaluation map, \emph{i.e.}, $\mathcal{C}_\mathcal{L}(\mathcal{P}, G) = \ev_\mathcal{P}(\mathcal{L}(G))$. Furthermore, we note that the evaluation map is $\F_q$-linear with kernel $\mathcal{L}(G - D)$, where $D \coloneqq P_1 + P_2 + \dots + P_n$ is a divisor of degree $n$. By this setting, we are ready to present the parameters of an AG code as follows.

\begin{theorem}[{\cite[Theorem 2.2.2]{stichtenoth2009algebraic}}]\label{thm:AG_code_dimension}
$\mathcal{C}_\mathcal{L}(\mathcal{P}, G)$ is an $[n, k, d]_q$ linear code with
\begin{equation*}
    k = \ell(G) - \ell(G - D) \quad \text{and} \quad d \geq n - \deg(G).
\end{equation*}
\end{theorem}

We can easily observe that as $\deg(G) < n$, the Riemann--Roch space $\mathcal{L}(G - D)$ will be the trivial space, and thus the dimension of the AG code is exactly $\ell(G)$. The following example shows how Reed--Solomon codes can be seen as a special case of AG codes.

\begin{example}[Reed--Solomon Codes as AG Codes]\label{ex:RS_as_AG_code}
Let us consider the rational function field $F = \F_q(x)$ and the pole $P_\infty$ of $x \in F$. For any nonnegative integer $k \leq n$, we know that $\mathcal{L}((k-1)P_\infty) = {\F_q[x]}_{< k}$. Now as $n \leq q$, we can consider $n$ distinct rational places $\mathcal{P} = \{P_{x - \alpha_1}, \dots, P_{x - \alpha_n}\}$ in $\P_F^1$ to denote points $\alpha_1, \ldots, \alpha_n$ in $\F_q$. Then the evaluation map yields a Reed--Solomon code, $\mathcal{C}_\mathcal{L}(\mathcal{P}, (k - 1)P_\infty)$, with parameters $[n, k]_q$. Note that the length $n$ and the dimension $k$ are the number of the rational places $P_{x - {\alpha}_i} \in \mathcal{P}$ and the dimension of the Riemann--Roch space $\mathcal{L}((k - 1)P_\infty)$, respectively.
\end{example}

In this paper, we are interested in one-point AG codes, which are codes coming from $\mathcal{L}(kP)$ with a rational place $P$ in a function field $F/K$. Note that Reed--Solomon codes presented in \cref{ex:RS_as_AG_code} are a special case of one-point AG codes.

Finally, we observe how the function field structure provides more information on the star product algebraic geometry codes by following \cite{couvreur2017cryptanalysis}. To this end, recall that the star product of two linear codes $\mathcal{C}, \mathcal{D} \subseteq \F_q^n$ is defined as
\begin{equation*}
    \mathcal{C} \star \mathcal{D} = \spn \{ c \star d \mid c \in \mathcal{C}, d \in \mathcal{D} \},
\end{equation*}
where $c \star d$ denotes the componentwise product. We study the star products of algebraic geometry codes with the following theorem. We define the product of two subspaces $V, W \subseteq F$ of a function field $F$ as
\begin{equation*}
    V \cdot W = \spn \{ fg \mid f \in V, g \in W \},
\end{equation*}
where $fg$ is the ordinary product of field elements.

\begin{theorem}\label{thm:product_of_RR_spaces}
Let $F/\F_q$ be a function field of genus $g$. Then, the product of Riemann--Roch spaces $\mathcal{L}(G)$ and $\mathcal{L}(H)$ satisfies
\begin{equation*}
    \mathcal{L}(G) \cdot \mathcal{L}(H) \subseteq \mathcal{L}(G + H).
\end{equation*}
Furthermore, equality holds if $\deg(G) \geq 2g$ and $\deg(H) \geq 2g + 1$.
\end{theorem}

The above theorem gives the following property for the star product of algebraic geometry codes.

\begin{remark}\label{rmk:star_product_of_AG_codes}
By applying \cref{thm:product_of_RR_spaces} for two AG codes $\mathcal{C}_\mathcal{L}(\mathcal{P}, G)$ and $\mathcal{C}_\mathcal{L}(\mathcal{P}, H)$ coming from the Riemann--Roch spaces $\mathcal{L}(G)$ and $\mathcal{L}(H)$ with $\deg(G) \geq 2g$ and $\deg(G) \geq 2g + 1$, we obtain
\begin{equation*}
    \mathcal{C}_\mathcal{L}(\mathcal{P}, G) \star \mathcal{C}_\mathcal{L}(\mathcal{P}, H) = \mathcal{C}_\mathcal{L}(\mathcal{P}, G + H).
\end{equation*}
\end{remark}

\begin{example}[Star Product of Reed--Solomon Codes] \label{ex:star_product_of_RS_codes}
Like in \cref{ex:RS_as_AG_code}, let us consider two Reed--Solomon codes $\mathcal{C}_\mathcal{L}(\mathcal{P}, (k - 1)P_\infty)$ and $\mathcal{C}_\mathcal{L}(\mathcal{P}, (k' - 1)P_\infty)$ with $1 \leq k \leq n$ and $2 \leq k' \leq n$. Since $g = 0$ for the rational function field we have $k - 1 \geq 2g $ and $k' - 1 \geq 2g + 1$. Finally, by \cref{rmk:star_product_of_AG_codes} we can see that
\begin{equation*}
    \mathcal{C}_\mathcal{L}(\mathcal{P}, (k - 1)P_\infty)\star \mathcal{C}_\mathcal{L}(\mathcal{P}, (k' - 1)P_\infty) = \mathcal{C}_\mathcal{L}(\mathcal{P}, (k + k' - 2)P_\infty).
\end{equation*}
Note that the evaluation map is not necessarily injective, therefore, one should note that the dimension of the Reed--Solomon code $\mathcal{C}_\mathcal{L}(\mathcal{P}, (k + k' - 2)P_\infty)$ is $\min\{k + k' - 1, n\}$.
\end{example}

\subsection{Linear SDMM}\label{sec:linear_SDMM}

In SDMM we consider the problem of distributing a matrix multiplication task to $N$ workers. The information contained in the matrices should be kept secret from the workers even if at most $X$ of them \emph{collude}, \emph{i.e.}, share their information in an attempt to infer information about the matrices. A general framework for linear SDMM schemes was presented in \cite{makkonen2022general}. The matrices $A$ and $B$ are partitioned to $mp$ and $np$ pieces 
\begin{equation*}
    A = \begin{pmatrix}
        A_{11} & \cdots & A_{1p} \\
        \vdots & \ddots & \vdots \\
        A_{m1} & \cdots & A_{mp}
    \end{pmatrix}, \quad
    B = \begin{pmatrix}
        B_{11} & \cdots & B_{1n} \\
        \vdots & \ddots & \vdots \\
        B_{p1} & \cdots & B_{pn}
    \end{pmatrix}
\end{equation*}
such that their product can be expressed as
\begin{equation*}
    AB = \begin{pmatrix}
        C_{11} & \cdots & C_{1n} \\
        \vdots & \ddots & \vdots \\
        C_{m1} & \cdots & C_{mn}
    \end{pmatrix},
\end{equation*}
where $C_{ik} = \sum_{j=1}^p A_{ij}B_{jk}$. In this paper we are interested in the case of $p = 1$, which is known as the \emph{outer product partitioning}.

The matrices $A$ and $B$ are encoded using linear codes $\mathcal{C}_A$ and $\mathcal{C}_B$ of length $N$ and dimensions $mp + X$ and $np + X$. In particular, if $G_A$ and $G_B$ are generator matrices of $\mathcal{C}_A$ and $\mathcal{C}_B$, respectively, then the encoded matrices are
\begin{align*}
    \tilde{A} &= (\tilde{A}_1, \dots, \tilde{A}_N) = (A_1, \dots, A_{mp}, R_1, \dots, R_X)G_A \\
    \tilde{B} &= (\tilde{B}_1, \dots, \tilde{B}_N) = (B_1, \dots, B_{np}, S_1, \dots, S_X)G_B,
\end{align*}
where $R_1, \dots, R_X$ and $S_1, \dots, S_X$ are matrices of suitable size whose entries are chosen uniformly at random from $\F_q$. The workers compute the products $\tilde{C}_i = \tilde{A}_i\tilde{B}_i$, which corresponds to the entries in the star product $\tilde{A} \star \tilde{B}$ consisting of componentwise matrix products.

\begin{definition}
A linear SDMM scheme is said to be \emph{decodable} if there are coefficients $\Lambda_i \in \F_q^{m \times n}$ that are independent of $A$, $B$ and the random matrices $R_k$ and $S_{k'}$ such that
\begin{equation*}
    AB = \sum_{i \in [N]} \Lambda_i \otimes \tilde{C}_i
\end{equation*}
for all $A$ and $B$.
\end{definition}

An SDMM scheme is said to be \emph{$X$-secure} if any $X$ colluding workers cannot extract any information about $A$ or $B$ from their shares. In terms of mutual information
\begin{equation*}
    I(\bm{A}, \bm{B}; \tilde{\bm{A}}_\mathcal{X}, \tilde{\bm{B}}_\mathcal{X}) = 0
\end{equation*}
for all $\mathcal{X} \subseteq [N]$ with $\lvert \mathcal{X} \rvert = X$. Here, the bold symbols correspond to the random variables of the nonbold symbols and $\tilde{\bm{A}}_\mathcal{X} = \{\tilde{\bm{A}}_i \mid i \in \mathcal{X}\}$ denotes the shares of the workers indexed by $\mathcal{X}$.

Let $\mathcal{C}_A^\security$ and $\mathcal{C}_B^\security$ denote the linear codes that are used to encode just the random part in the linear SDMM scheme, \emph{i.e.}, the codes spanned by the rows of the generator matrices $G_A$ and $G_B$ that correspond to the random matrices. One way to show the security of the linear SDMM scheme is by the following theorem. The proof of the theorem uses standard arguments from information theory and relies on the fact that all the $X \times X$ submatrices of the generator matrices are invertible.

\begin{theorem}[{\cite[Theorem 1]{makkonen2022general}}]\label{thm:linear_SDMM_security}
A linear SDMM scheme is $X$-secure if $\mathcal{C}_A^\security$ and $\mathcal{C}_B^\security$ are MDS codes.
\end{theorem}

To motivate our construction we consider the following example of linear SDMM based on the GASP scheme presented in \cite{d2020gasp}.

\begin{example}[GASP]\label{ex:GASP}
The matrices $A$ and $B$ are split into $m = n = 3$ submatrices using the outer product partitioning. We wish to protect against $X = 2$ colluding workers. Define the polynomials
\begin{align*}
    f(x) &= A_1 + A_2x + A_3x^2 + R_1x^9 + R_2x^{12}, \\
    g(x) &= B_1 + B_2x^3 + B_3x^6 + S_1x^9 + S_2x^{10},
\end{align*}
where $R_1, R_2, S_1, S_2$ are matrices of appropriate size that are chosen uniformly at random over $\F_q$. The exponents are chosen carefully so that the total number of workers needed is as low as possible. Let $\alpha_1, \dots, \alpha_N \in \F_q^\times$ be distinct nonzero points and evaluate the polynomials $f(x)$ and $g(x)$ at these points to get the encoded matrices
\begin{equation*}
    \widetilde{A}_i = f(\alpha_i), \quad \widetilde{B}_i = g(\alpha_i).
\end{equation*}
The $i$th encoded matrices are sent to the $i$th worker node. The workers compute the matrix products $\widetilde{C}_i = \widetilde{A}_i \widetilde{B}_i$ and send these  to the user. The user receives evaluations of the polynomial $h(x) = f(x)g(x)$ from each worker. Using the definition of $f(x)$ and $g(x)$ we can write out the coefficients of $h(x)$ as
\begin{align*}
    h(x) &= A_1B_1 + A_2B_1x + A_3B_1x^2 + A_1B_2x^3 + A_2B_2x^4 \\ 
    &+ A_2B_3x^5 + A_1B_3x^6 + A_2B_3x^7 + A_3B_3x^8 \\
    &+ (\text{terms of degree $\geq 9$}).
\end{align*}
We notice that the coefficients of the first $9$ terms are exactly the submatrices we wish to recover. We may study the coefficients that appear in the product using the \emph{degree table} in \cref{tab:degree_table}. As the degree table contains 18 distinct elements, we need 18 responses from the workers, provided that the corresponding linear equations are solvable. In this case, the number of workers is $N = 18$.

The general choice of the exponents in the polynomials $f(x)$ and $g(x)$ is explained in \cite{d2021degree}. The security of the scheme is proven by showing that the condition of \cref{thm:linear_SDMM_security} is satisfied for a suitable choice of evaluation points.

\begin{table}[t]
    \centering
    \begin{tabular}{C|CCC|CC}
           &  0 &  3 &  6 &  9 & 10 \\ \hline
         0 &  0 &  3 &  6 &  9 & 10 \\
         1 &  1 &  4 &  7 & 10 & 11 \\
         2 &  2 &  5 &  8 & 11 & 12 \\ \hline
         9 &  9 & 12 & 15 & 18 & 19 \\
        12 & 12 & 15 & 18 & 21 & 22
    \end{tabular}
    \caption{Degree table of the GASP scheme}
    \label{tab:degree_table}
\end{table}
\end{example}

It may be desirable to recover the result from a subset of the workers such that \emph{stragglers}, \emph{i.e.}, slow or unresponsive workers, do not negatively affect the computation time. The minimal number of responses needed to decode the result in the worst case scenario is known as the \emph{recovery threshold}. The recovery threshold of the scheme described in the above example is 18, which equals the number of workers. This means that the scheme is not able to tolerate stragglers.

\section{Using Algebraic Geometry Codes in Linear SDMM}\label{sec:AG_codes_in_SDMM}

Extending \cref{ex:GASP} with algebraic geometry codes is based on the fact that the exponents of the monomials $x^i$ can also be chosen with respect to some pole numbers of the pole of $x$ in the rational function field $\F_q(x)$. This approach leads us to set up a pole number table that is similar to the degree table in \cref{tab:degree_table}. Furthermore, we consider certain subcodes of AG codes, which is a neat generalization of the subcodes of Reed--Solomon codes used in \cref{ex:GASP}.

We consider one-point algebraic geometry codes from a divisor of the form $G = kP$ for some $P \in \P_F^1$, and subcodes of such codes.  This corresponds to the extension of the GASP codes, where the encoding is done by certain subcodes of Reed--Solomon codes. In particular, we choose some pole numbers in $W(P)$, say
\begin{equation*}
    \phi = (\phi_1, \dots, \phi_{m + X}), \quad \gamma = (\gamma_1, \dots, \gamma_{n + X}).
\end{equation*}
Additionally, we choose functions $f_1, \dots, f_{m + X}$ and $g_1, \dots, g_{n + X}$ such that $(f_j)_\infty = \phi_j P$ and $(g_{j'})_\infty = \gamma_{j'} P$. The matrices are partitioned according to the partition in \cref{sec:linear_SDMM} with $p = 1$, \emph{i.e.}, the outer product partition. We encode our matrices by computing the following linear combinations
\begin{align*}
    f &= \sum_{k=1}^X R_k f_k + \sum_{j=1}^m A_j f_{X + j}, \\
    g &= \sum_{k'=1}^X S_{k'} g_{k'} + \sum_{j'=1}^n B_{j'} g_{X + j'}. 
\end{align*}
We want the encodings to be injective (by \cite[Proposition 3]{makkonen2022general}), which means that the collection $f_1, \dots, f_{m + X}$ has to be linearly independent. This can be achieved, for instance, by choosing the pole numbers $\phi_1, \dots, \phi_{m + X}$ to be distinct. Similarly, we want the pole numbers $\gamma_1, \dots, \gamma_{n + X}$ to be distinct. The goal is to design the pole numbers $\phi$ and $\gamma$ such that we can extract the submatrix products $A_jB_{j'}$ from the product $h = fg$. We may study the possible pole numbers that appear in the product by considering the \emph{pole number table} $\phi \oplus \gamma$, which is defined as the outer sum of $\phi$ and $\gamma$. Recall that the pole divisor of the product $f_jg_{j'}$ is $(f_jg_{j'})_\infty = (\phi_j + \gamma_{j'})P$. This is a direct generalization of the degree tables introduced in \cite{d2021degree}.

\begin{remark}
Notice that our encoding is done such that the random matrices correspond to the functions $f_1, \dots, f_X$ and $g_1, \dots, g_X$ and the matrix partitions correspond to the functions $f_{m + 1}, \dots, f_{m + X}$ and $g_{n + 1}, \dots, g_{n + X}$. This is merely a notational difference to earlier constructions, but should be noted in the next section where the pole number table is analyzed.
\end{remark}

\section{Construction of the PoleGap Scheme}\label{sec:construction}

Let $\F_q$ be a finite field with odd characteristic. Fix the partitioning parameters $m \geq 2$ even, $n \geq 1$ and the collusion parameter $X \geq 1$, and set
\begin{equation*}
    d = m(n - 1) + 2X - 1.
\end{equation*}
Let $f(x) = \prod_{i=1}^d (x - \alpha_i) \in \F_q[x]$ be a square-free polynomial of degree $d$ and consider the Kummer extension $F/\F_q$ defined by $y^2 = f(x)$. By \cref{thm:Kummer_special} this function field has genus
\begin{equation*}
    g = \frac{d - 1}{2} = \frac{1}{2}\left(m(n-1) + 2X - 2\right).
\end{equation*}
Let $P_\infty$ be the unique pole of $x \in F$. Then we have that $(x)_\infty = 2P_\infty$. Hence, the Weierstrass semigroup at $P_\infty$ is
\begin{equation*}
    W(P_\infty) = \{ 0, 2, 4, \dots, 2g, 2g + 1, \dots \},
\end{equation*}
\emph{i.e.}, all even integers are pole numbers, as well as all integers at least $2g$.

Consider the following sequences of natural numbers
\begin{align*}
    \phi &= (0, 2, 4, \dots, 2X - 2, d, d + 1, \dots, d + m - 1) \in \N_0^{m + X} \\
    \gamma &= (0, 2, 4, \dots, 2X - 2, m + 2X - 2, 2m + 2X - 2, \dots, mn + 2X - 2) \in \N_0^{n + X}.
\end{align*}
We see that all the elements of $\gamma$ are even as $m$ is assumed to be even, so all the elements of $\gamma$ are pole numbers. Similarly, the first $X$ elements of $\phi$ are even, so they are pole numbers. Finally, the last $m$ elements of $\phi$ are at least $d = 2g + 1$, so they are also pole numbers. Therefore, we can choose functions $f_1, \dots, f_{m + X} \in F$ and $g_1, \dots, g_{n + X} \in F$ such that $(f_j)_\infty = \phi_j P_\infty$ and $(g_{j'})_\infty = \gamma_{j'} P_\infty$. We will consider the subspaces
\begin{equation*}
    L_A = \spn \{f_1, \dots, f_{m+X}\}, \quad L_B = \spn \{g_1, \dots, g_{n+X}\}.
\end{equation*}
The maximal pole number of $fg$ for $f \in L_A$ and $g \in L_B$ is $\phi_{m + X} + \gamma_{n + X} = 2mn + 4X - 4$. Hence, we consider the divisor $G = (2mn + 4X - 4)P_\infty$. We see that $L_A, L_B \subseteq \mathcal{L}(G)$ and $L_A \cdot L_B \subseteq \mathcal{L}(G)$ by choice of $G$.

Let $\hat{\mathcal{P}} \subseteq \P_F^1 \setminus \{P_\infty\}$ be a set of places of size $\hat{N} = \lvert \hat{\mathcal{P}} \rvert$. By considering the evaluation map $\ev_{\hat{\mathcal{P}}} \colon \mathcal{L}(G) \to \F_q^{\hat{N}}$ we define the codes $\hat{\mathcal{C}}_A = \ev_{\hat{\mathcal{P}}}(L_A)$ and $\hat{\mathcal{C}}_B = \ev_{\hat{\mathcal{P}}}(L_B)$. The star product of these codes is
\begin{equation*}
    \hat{\mathcal{C}}_A \star \hat{\mathcal{C}}_B = \ev_{\hat{\mathcal{P}}}(L_A \cdot L_B) \subseteq \ev_{\hat{\mathcal{P}}}(\mathcal{L}(G)) = \mathcal{C}_\mathcal{L}(\hat{\mathcal{P}}, G).
\end{equation*}
Assuming that $\hat{N} > \deg G = 2mn + 4X - 4$ we get
\begin{align*}
    \dim \hat{\mathcal{C}}_A \star \hat{\mathcal{C}}_B &\leq \dim \mathcal{C}_\mathcal{L}(\hat{\mathcal{P}}, G) \\
    &= \deg G + 1 - g \\
    &= \frac{3}{2}mn + \frac{1}{2}m + 3X - 2.
\end{align*}
This follows from the fact that
\begin{equation*}
    \deg G = 2mn + 4X - 4 \geq m(n - 1) + 2X - 3 = 2g - 1.
\end{equation*}

Let $\mathcal{P} \subseteq \hat{\mathcal{P}}$ be a subset of places corresponding to an information set of $\hat{\mathcal{C}}_A \star \hat{\mathcal{C}}_B$. Define the codes $\mathcal{C}_A = \ev_\mathcal{P}(L_A)$ and $\mathcal{C}_B = \ev_\mathcal{P}(L_B)$. These codes have length $N = \dim \hat{\mathcal{C}}_A \star \hat{\mathcal{C}}_B$ and dimensions $m + X$ and $n + X$, respectively.

We partition the matrices $A$ and $B$ using the outer product partitioning such that
\begin{equation*}
    A = \begin{pmatrix}
        A_1 \\ \vdots \\ A_m
    \end{pmatrix}, \quad
    B = \begin{pmatrix}
        B_1 & \cdots & B_n
    \end{pmatrix}.
\end{equation*}
Then the product $AB$ can computed directly from the submatrix products $A_jB_{j'}$. For encoding $A$ and $B$ we consider the functions
\begin{align*}
    f &= \sum_{k=1}^X R_k f_k + \sum_{j=1}^m A_j f_{X+j} \in L_A, \\
    g &= \sum_{k'=1}^X S_{k'} g_k + \sum_{j'=1}^n B_{j'} g_{X+j'} \in L_B.
\end{align*}
The encoded pieces of $A$ and $B$ are determined by $\tilde{A} = \ev_\mathcal{P}(f)$ and $\tilde{B} = \ev_\mathcal{P}(g)$. As each worker computes the product of their encoded pieces, the user receives the responses $\ev_\mathcal{P}(h)$, where $h = fg$. As the set $\mathcal{P}$ is chosen to be an information set of the code $\hat{\mathcal{C}}_A \star \hat{\mathcal{C}}_B$, we get that $\ev_\mathcal{P}$ is injective on $L_A \cdot L_B$. Hence, we can recover $h$ from the response vector $\ev_\mathcal{P}(h)$.

\begin{table*}[t]
    \scriptsize
    \centering
    \resizebox{\textwidth}{!}{\begin{tabular}{C|CCC|CCC}
        & 0 & \cdots & 2X - 2 & m + 2X - 2 & \cdots & mn + 2X - 2 \\ \hline
        0 & 0 &  \dots & 2X - 2 & m + 2X - 2 & \cdots & mn + 2X - 2 \\
        \vdots & \vdots & \ddots & \vdots & \vdots & \ddots & \vdots \\
        2X - 2 & 2X - 2 & \cdots & 4X - 4 & m + 4X - 4 & \cdots & mn + 4X - 4 \\ \hline
        m(n - 1) + 2X - 1 & m(n - 1) + 2X - 1 & \dots & m(n - 1) + 4X - 3 & mn + 4X - 3 & \dots & 2mn - m + 4X - 3 \\
        \vdots & \vdots & \ddots & \vdots & \vdots & \ddots & \vdots \\
        mn + 2X - 2 & mn + 2X - 2 & \cdots & mn + 4X - 4 & mn + m + 4X - 4 & \cdots & 2mn + 4X - 4
    \end{tabular}}
    \caption{The pole number table of the chosen pole numbers $\phi$ and $\gamma$ with entry $(j, j')$ corresponding to $\phi_j + \gamma_{j'}$. Furthermore, we used the fact that $d = m(n - 1) + 2X - 1$.}
    \label{tab:pole_number_table}
\end{table*}

Let us study the form of the product $h = fg$ with the pole number table depicted in \cref{tab:pole_number_table}. The pole numbers in the bottom right quadrant are all at least $mn + 4X - 3$, while the pole numbers in the remaining quadrants are at most $mn + 4X - 4$. Furthermore, the pole numbers in the bottom right quadrant are all distinct. The product $h = fg \in L_A \cdot L_B$ has the form
\begin{align*}
    h &= \sum_{k=1}^X \sum_{k'=1}^X R_k S_{k'} f_k g_{k'} + \sum_{k=1}^X \sum_{j'=1}^n R_k B_{j'} f_k g_{X + j'} \\
    &+ \sum_{j=1}^m \sum_{k'=1}^X A_j S_{k'} f_{X + j} g_{k'} + \sum_{j=1}^m \sum_{j'=1}^n A_j B_{j'} f_{X + j} g_{X + j'}.
\end{align*}
The terms in the first three sums correspond to the pole numbers in the first three quadrants of the pole number table, \emph{i.e.}, they are contained in the Riemann--Roch space $\mathcal{L}(G')$, where $G' = (mn + 4X - 4)P_\infty$. Therefore, we consider the decomposition $L_A \cdot L_B = L \oplus H$, where $L = (L_A \cdot L_B) \cap \mathcal{L}(G')$ and $H = \spn \{f_{X + j}g_{X + j'} \mid j \in [m], j' \in [n]\}$. We can therefore project $h$ to $H$ to obtain
\begin{equation*}
    \sum_{j=1}^m \sum_{j'=1}^n A_jB_{j'} f_{X + j}g_{X + j'}.
\end{equation*}
The functions $f_{X + j}g_{X + j'}$ have distinct pole numbers, which means that they are linearly independent. We can therefore compute the products $A_jB_{j'}$ for all pairs $(j, j') \in [m] \times [n]$.

\begin{theorem}
By choosing $f_j = x^{j-1}$ and $g_j = x^{j-1}$ for $1 \leq j \leq X$ and the set $\hat{\mathcal{P}}$ appropriately, the scheme is $X$-secure.
\end{theorem}

\begin{proof}
Let $\alpha \in \F_q$. There can be at most two places in $\P_F$ such that $(x - \alpha)(P) = 0$, since the extension degree of $F$ over $K(x - \alpha) = K(x)$ is 2. Choose the set $\hat{\mathcal{P}} \subseteq \P_F^1 \setminus \{P_\infty\}$ such that it includes at most one of these two solutions for all $\alpha \in \F_q$. Thus, $\hat{\mathcal{P}}$ can be chosen to have size at least $\frac{1}{2}\lvert \P_F^1 \setminus \{P_\infty\} \rvert$. Then $x(P) \in \F_q$ are distinct for all $P \in \hat{\mathcal{P}}$. By the Hasse--Weil bound we can take $\P_F^1$ to be as large as we want for large enough $q$.

The choice of $f_j$ and $g_j$ is appropriate, since $(x^{j-1})_\infty = 2(j-1)P_\infty$, which is exactly what was required in the construction of the scheme. Let $\mathcal{P} = \{P_1, \dots, P_N\}$ and $\alpha_i = x(P_i) \in \F_q$. The generator matrix of the security part is then
\begin{equation*}
    G_A = \begin{pmatrix}
        f_1(P_1) & \cdots & f_1(P_N) \\
        \vdots & \ddots & \vdots \\
        f_X(P_1) & \cdots & f_X(P_N)
    \end{pmatrix}
    = \begin{pmatrix}
        1 & \cdots & 1 \\
        \alpha_1 & \cdots & \alpha_N \\
        \vdots & \ddots & \vdots \\
        \alpha_1^{X - 1} & \cdots & \alpha_N^{X - 1}
    \end{pmatrix}.
\end{equation*}
This matrix is the generator matrix of an MDS code, since it is a Vandermonde matrix with distinct evaluation points. Similarly, $G_B = G_A$ is the generator matrix of an MDS code. By \cref{thm:linear_SDMM_security} the scheme is $X$-secure.
\end{proof}

We conclude this section by formulating our construction as the following theorem. We call this the PoleGap SDMM scheme.

\begin{theorem}[PoleGap SDMM scheme]\label{thm:construction}
Given partitioning parameters $m, n \geq 1$ such that at least one of them is even, and the collusion parameter $X \geq 1$, our construction gives an $X$-secure linear SDMM scheme that uses
\begin{equation*}
    N \leq \begin{dcases*}
        \frac{3}{2}mn + \frac{1}{2}m + 3X - 2 & if $m$ is even and $n$ is odd \\
        \frac{3}{2}mn + \frac{1}{2}n + 3X - 2 & if $m$ is odd and $n$ is even \\
        \frac{3}{2}mn + \frac{1}{2}\min\{m, n\} + 3X - 2 & if $m$ and $n$ are even
    \end{dcases*}
\end{equation*}
workers.
\end{theorem}

\section{Comparison}\label{sec:comparison}

In this section, we will compare our construction to some state-of-the-art schemes which use the same matrix partitioning that is used in our construction. Comparing schemes with the same partitioning makes the comparison meaningful since the parameters in different partitioning methods are not directly comparable. Some of the state-of-the-art SDMM schemes which use the outer product partitioning include the A3S scheme in \cite{kakar2019capacity} and the GASP scheme in \cite{d2021degree}. We will compare the schemes by computing their rate
\begin{equation*}
    \mathcal{R} \coloneqq \frac{mn}{N},
\end{equation*}
which is inversely proportional to the number of workers that are needed. The goal is to maximize the rate. The number of workers in the A3S is $N = (m + X)(n + 1) - 1 = mn + m + (n + 1)X - 1$. The number of workers in the GASP scheme was explicitly computed in \cite{d2021degree}, but an upper bound is given by the GASP\textsubscript{big} scheme, which uses $N \leq 2mn + 2X - 1$ workers. 

\begin{figure*}[t]
    \centering
    \includegraphics[width=0.8\textwidth]{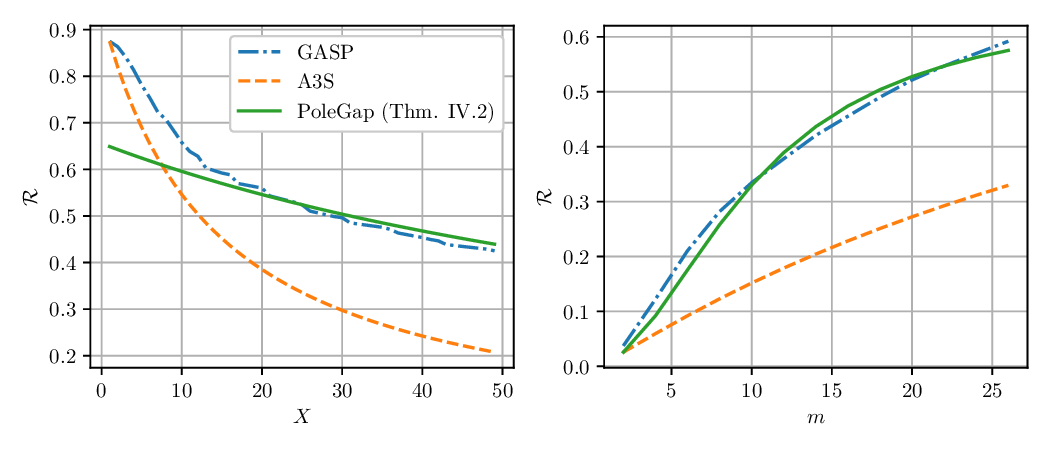}
    \caption{Comparison of the rate $\mathcal{R}$ between the GASP scheme in \cite{d2021degree}, the A3S scheme in \cite{kakar2019capacity} and our PoleGap scheme in \cref{thm:construction}. The left plot has fixed $m = n = 14$ and the rate is measured as a function of $X$. The right plot has fixed $X = 50$ and the rate as a function of even $m = n$.}
    \label{fig:rate_comparison}
\end{figure*}

Heuristically, we see that for fixed $m, n \geq 3$, our scheme achieves a better rate than the A3S scheme as $X \to \infty$. Similarly, for a fixed $X$, our scheme achieves a better rate as $mn \to \infty$ compared to the upper bound of the GASP\textsubscript{big} scheme. In \cref{fig:rate_comparison} we see that for some parameters our construction achieves a higher rate than the general GASP scheme. In particular, our scheme achieves a better rate than the A3S or the GASP scheme for 9.1\% of parameters in the range $2 \leq m \leq 50$, $1 \leq n \leq m$, $1 \leq X \leq 50$ and $m$ even. The A3S scheme seems to perform significantly worse compared to the two other schemes in terms of rate. However, it is simple to add straggler tolerance to the A3S, while this is not simple for the other schemes due to their more complicated choices of evaluation points.

These statistics and \cref{fig:rate_comparison} show that our construction is a slight improvement over the state-of-the-art schemes for some parameters. However, we believe that this construction shows that using algebraic geometry codes in SDMM has great potential.

\begin{remark}\label{rmk:field_size}
Our construction only requires the underlying finite field to have odd characteristic and suitably many rational places in the Kummer extension, which can be satisfied if the field is large enough. This allows us to utilize suitably large prime fields, which are favorable for meaningful computation; namely, the initial data is often real-valued and embedded into a suitable finite field. For the computation results to be understandable in the initial domain, we need to have some homomorphic relation between the real numbers and the chosen finite field. For instance, truncating and embedding real data to a field extension with a small characteristic severely fails in terms of this requirement. For example, let $A \in \F_q^{t \times s}$ and $B \in \F_q^{s \times r}$ be the all-ones matrices with $s = \chr(\F_q)$. Then $AB = 0 \in \F_q^{t \times r}$, but it is now ambiguous what real-valued matrix this corresponds to since the product of the all-ones matrices over the real numbers is nonzero. Unlike in distributed storage or communications, there is no need to minimize the field size, since the computations require us to work with a field with large characteristic.
\end{remark}

Given the pole number table in \cref{tab:pole_number_table} we could also do a similar construction as the GASP scheme in \cite{d2020gasp}. Instead of choosing suitable functions with certain pole numbers we could choose monomials with specified degrees. In this approach the resulting star product code would not be an algebraic geometry code, but would some arbitrary linear code. The easy computation of the dimension of such a code by utilizing the genus of the underlying function field, as well as being able to utilize error correction algorithms, makes this approach quite interesting.

\section{Conclusions}

In this paper, we investigated how algebraic geometry codes may be used in secure distributed matrix multiplication by giving an explicit construction coming from Kummer extensions, which we call the PoleGap SDMM scheme. This construction is a proof-of-concept that using AG codes can bring advantages over previously considered schemes. The gaps in the Weierstrass semigroup correspond to the gaps in the GASP scheme, which means that the number of the gaps is counted by the genus of the algebraic function field. Additionally, using AG codes may have other benefits, such as more flexibility in field size and error correction. The fact that the response code is known to be an AG code allows us to utilize error correction algorithms such as list decoding for the SDMM scheme. As future work, we would also like to extend our construction to the analog case, since the theory of Kummer extensions works just as well over the fields $\R$ or $\C$.

\section*{Acknowledgements}

This work has been supported by the Academy of Finland under Grant No.\ 336005. The work of O.~Makkonen is supported by the Vilho, Yrjö and Kalle Väisälä Foundation of the Finnish Academy of Science and Letters. 

\bibliography{bib.bib}
\bibliographystyle{ieeetr}

\end{document}